\documentclass{article}
\usepackage[utf8]{inputenc}
\usepackage{amsmath,amsthm,amssymb}
\usepackage{color}
\usepackage{bm,bbm}
\usepackage{ascmac,paralist}

\usepackage{mathtools}
\usepackage{comment}

\usepackage[sorting=none]{biblatex}
\newtheorem{theorem}{Theorem}[section]
\newtheorem{definition}[theorem]{Definition}
\newtheorem{corollary}[theorem]{Corollary}
\newtheorem{lemma}[theorem]{Lemma}
\newtheorem{proposition}[theorem]{Proposition}

\newcommand{\inpr}[2]{\left\langle #1 ,\,  #2 \right\rangle}
\renewcommand{\epsilon}{\varepsilon}

\newcommand{\HC}{\mathcal{H}}

\newcommand{\Lr}[1]{\left\{#1\right\}}
\newcommand{\LR}[1]{\left[#1\right]}
\renewcommand{\exp}[1]{{\rm exp}\LR{#1}}

\newcommand{\bs}[1]{\boldsymbol{#1}}
\makeatletter
\newtheorem*{rep@theorem}{\rep@title}
\newcommand{\newreptheorem}[2]{%
\newenvironment{rep#1}[1]{%
 \def\rep@title{#2 \ref{##1}}%
 \begin{rep@theorem}}%
 {\end{rep@theorem}}}
\makeatother
\newreptheorem{theorem}{Theorem}

\usepackage{ulem}

\title{A method of approximation of discrete Schr{\"o}dinger  equation with the normalized Laplacian by discrete-time quantum walk on graphs}

\author{Kei Saito$^1$, Etsuo Segawa$^2$\\
$^1$
{\small Department of Information Systems Creation, Faculty of Informatics, Kanagawa University, 
} \\
{\small Kanagawa, Yokohama, 221-8686, Japan}
\\
$^2${\small Graduate School of Environment and Information Sciences, Yokohama National University,}\\ {\small Hodogaya, Yokohama 240-8501, Japan
}
}
\date{}

\begin{document}
\maketitle
\begin{abstract}
We propose a class of continuous-time quantum walk models on graphs induced by a certain class of discrete-time quantum walk models with a parameter $\epsilon\in [0,1]$. 
Here the graph treated in this paper can be applied both finite and infinite cases. 
The induced continuous-time quantum walk is an extended version of the (free) discrete-Schr{\"o}dinger equation driven by the normalized Laplacian: the element of the weighted Hermitian takes not only a scalar value but also a matrix value depending on the underlying discrete-time quantum walk. We show that each discrete-time quantum walk with an appropriate setting of the parameter $\epsilon$ in the long time limit identifies with its induced continuous-time quantum walk and give the running time for the discrete-time to approximate the induced continuous-time quantum walk with a small error $\delta$. 
We also investigate the detailed spectral information on the induced continuous-time quantum walk.

\end{abstract}
\section{Introduction}
In the research field of quantum walks, finding the connection between continuous-time quantum walk and discrete-time quantum walk is one of the natural and interesting problems~for examples, \cite{Strauch2006,Strauch2007,Childs2010,Shikano2013,Molfetta2012,DhBr,ManKon}. 
Strauch established the connection between them on the one-dimensional lattice as the underlying graph~\cite{Strauch2006,Strauch2007}. 
In \cite{Shikano2013}, the crossover of the limit distributions between the continuous- and discrete-time quantum walks. 
Childs gave a method to simulate the continuous-time quantum walk driven by arbitrary Hamiltonian operator by using an induced discrete-time quantum walk \cite{Childs2010}. 
In \cite{DhBr}, the connection of them on general connected graph $G=(V,A)$, where $A$ is the set of the symmetric arcs, is obtained by a heuristic argument and the induced continuous-time quantum walk is applied to the quantum search.  
In the methods of \cite{Childs2010,DhBr}, some graph deformations are needed to approximate the continuous-time quantum walk. 
In this study, we attempt to obtain a continuous-time quantum walk on a graph whose Hilbert space is $\ell^2(V)$ from a discrete-time quantum walk on the same graph whose Hilbert space is $\ell^2(A)$, and also attempt to show how the discrete-time quantum walk approximates the corresponding continuous-time quantum walk without any deformation of the underlying graph, rigorously. 

To explain our idea, let us give a quick review on the well-known continuous-time random walk on the graph. 
In a typical construction of the continuous-time random walk on a connected graph $G=(V, E)$, an independent and unit rate Poisson process is assigned at each edge and the following procedure is repeated. 
If a random walker exists at the vertex $u\in V$ at time $t$, the next jump to a neighbor occurs at the first time with $s>t$ at which there is an increment of the Poisson process on some edge incident on $u$, say $\{u,w\}\in E$, and moves to the vertex $w$. 
Let $\{T_j\}_{j=1,\dots,\deg(v)}$ be i.i.d. random variables following the exponential distribution $\mathrm{Ex}(1)$ which correspond to the waiting times of the Poisson bells assigned at edges incident on $v$ whose degree is $\deg(v)$.
Since $P(\min\{T_1,\dots,T_{\deg(v)}\}=T_j)=1/{\deg(v)}$
, roughly speaking, this continuous-time random walk can be regarded as the isotropic discrete-time random walk on the same graph. 
This is one of the simplest connection between continuous-time random walk and discrete-time random walk. 
However this observation of the random walk can not be directly reflected as an analogues to the quantum walk property because the quantum walk are not probability processes. 
On the other hand, the probability, that the Poisson bell on the edge $\{v,w\}$ rings during a small time interval $\Delta$, and that no other bell on an edge incident on $v$ does so, is 
$(1-e^{-\Delta})\times (e^{-\Delta})^{\deg(v)}\sim \Delta$. This means that the event of the moving of a particle is not so frequent in the continuous-time random walk. In this paper, we have eyes on this property of the continuous-time random walk as an analogues subject to the quantum walk case and construct a discrete-time quantum walk so that it implements a continuous-time quantum walk in some limit. 

Let us explain it briefly. 
The time evolution of a typical discrete-time quantum walk on graph $G=(V,A)$, where $A$ is the set of symmetric arcs induced by edge set $E$, is a unitary operator on $\ell^2(A)$ and described by $U_o=S_oC$.
Here $S_o$ is called the flip flop shift operator, such that for every standard base $\delta_a$ $(a\in A)$,  $S_o\delta_{a}=\delta_{\bar{a}}$, where $\bar{a}$ is the inverse arc of $a$, and $C$ acts as the local unitary operator on each subspace $\mathbb{C}^{X_u}$ for every vertex $u$, where $X_u=\left\{a\in A \;|\; t(a)=u\right\}$
and $t(a)$ means the terminus of the arc $a$. 
Such a local unitary operator which describes the local scattering at each  vertex $u$ is called the local coin operator assigned at vertex $u$. 
Our idea is that by extending the shift operator so that 
$S = \sqrt{1-\epsilon^2}I_{\mathcal{A}} + i\epsilon S_o$
with a small parameter $0<\epsilon\ll 1$, an analogues situation of a rare event of the moving to the neighbor in the continuous-time random walk is created. We call the parameter $\epsilon$ a mobility parameter. 
Here $I_X$ means the identity operator on a vector space $X$.
In the setting of \cite{Childs2010}, by adding the self-loops to all the vertices, the analogues situation of the ``lazyness" staying at the same vertex of the random walk is created. Such a graph deformation causes the enlargement of the Hilbert space of the induced discrete-time quantum walk. 
On the other hand, our case dose not need such an enlargement of the Hilbert space.    
When the graph $G$ is the one-dimensional lattice, the walk corresponds to the split-step quantum walk introduced by Kitagawa et al. \cite{Kitagawa2010, Kitagawa2012}.

In the setting of such a time evolution operator of the quantum walk $U(\epsilon)=SC$, we obtain the following main theorem. See section~3 for the proof. 
\begin{theorem}\label{thm:main}
Let $G=(V,A)$ be the underlying graph which is simple and connected. 
Assume the spectrum of the coin operator $C$ is $\{\pm 1\}$. 
Let $\psi_{N}^{[D,\epsilon]}\in \ell^2(A)$ be the state at a final time $N\in \mathbb{N}$ of the discrete-time quantum walk with the time evolution operator $U^2(\epsilon/2)$ and the initial state $\varphi_0\in \mathcal{A}$, that is, 
\[ \psi_0^{[D,\epsilon]}=\varphi_0;\;\;\psi_{n}^{[D,\epsilon]}=U^2(\epsilon/2)\;\psi_{n-1}^{[D,\epsilon]} \;\text{ $(n=1,\dots,N)$}, \]
while $\phi_{t}^{[C]}\in \ell^2(A)$ be the state of the continuous-time quantum walk with the Hermitian operator $H=(S_o+S_oCS_o)/2$ at time $t\in \mathbb{R}_{>0}$ and the same initial state $\varphi_0$, that is,
\[ \phi_0^{[C]}=\varphi_0;\;\; -i\frac{\partial}{\partial t}\phi_t^{[C]} = H\phi_t^{[C]} \;\text{ $(t>0)$}. \]
Then we have 
\begin{equation*}
\phi_t^{[C]}(a) = \lim_{N\to\infty} \psi^{[D,\;t/N]}_N(a)
\end{equation*}
for any $a\in A$ and 
$t>0$.
\end{theorem}
Our main theorem implies that 
this continuous-time quantum walk, say the continuous-time Szegedy walk, at time $t$ can be approximated by the discrete-time quantum walk with the mobility parameter $\epsilon=t/N$ at time $N$ for large $N\gg 1$. 
Then the continuous-time Szegedy walk of the unit time is created by the long time and small mobility parameter's limits of this  discrete-time quantum walk. 
The following corollary shows that the continuous-time quantum walk on $\ell^2(V)$ driven by the normalized Laplacian (discrete Schr{\"o}dinger equation) can be reproduced. 
This proof is immediately obtained by 
combining Theorem~\ref{thm:main} with Proposition~\ref{prop:DiscSch}. 
\begin{corollary}\label{cor:discrete-Sch}
Let the local quantum coin be the Grover's matrix.
Set $d:\ell^2(A)\to \ell^2(V)$ such that 
$(d\psi)(u)=\sum_{t(a)=u}1/\sqrt{\deg(t(a))} \psi(a)$. 
Let $T$ be the self adjoint operator on $\ell^2(V)$ such that 
$(Tf)(u)=\sum_{t(a)=u}1/\sqrt{\deg(o(a))\deg(t(a))}\;f(o(a))$
 for any 
$f\in \ell^2(V)$
and $u\in V$, where $o(a)$ means the origin of an arc $a$.
Consider the following continuous-time quantum walk on $\ell^2(V)$ with the normalized Laplacian $T-I$
with initial state $g\in \ell^2(V)$:  
\[ -i\frac{\partial}{\partial t}f_t= (T-I) f_t,\;f_0=g. \]
Then we have 
\[ f_t=e^{-it}\;d \lim_{N\to\infty} \psi^{[D,\;t/N]}_N,   \]
where $\psi^{[D,\;t/N]}_0=d^*g$ for any $t>0$. 
\end{corollary}
We emphasize that the continuous-time Szegedy walk has the invariant subspace under the action of the self adjoint operator $H$ which can reproduce the above (free) discrete-Schr\"odinger equation but also has an additional invariant subspace. 
In this paper, we also clarify that this invariant subspace is the same as the subspace which gives the localization of the discrete-time Grover walk if the underling graph is infinite and has a closed cycle~\cite{HKSS}. See Theorem~\ref{thm:spectrum}: in this case, the eigenspaces $\mathcal{B}_{\pm}$ in Theorem~\ref{thm:spectrum} are equivalent to those of the Grover walk. 

The rest of this paper is organized as follows. 
In section~2, we introduce the setting of graph and discrete-time and continuous-time quantum walks, namely, the discrete-time and continuous-time Szegedy walk. respectively. Throughout this paper, the walk whose time evolution operator is described by two distinct involution operators are called the Szegedy walk. 
Indeed, the both continuous- and discrete-time quantum walks treated here are constructed by the flip flop shift $S_o$ and coin $C$ which are involution, that is, $S_o^2=C^2=I$. 
In section~3, we give the proof of the main theorem and show the running time of the discrete-time quantum walk to reproduce the state of the corresponding continuous-time quantum walk with a small error. 
In section~4, we analyze the detailed spectral information on the continuous-time Szegedy walk.

\section{Definitions and models}
\subsection{Setting of graph}
In this study, we treat simple and connected {\it symmetric digraph} $G=(V,A)$, that is,  
$a\in A$ if and only if $\bar{a}\in A$, where $\bar{a}$ is the inverse arc of $a$. 
The origin and terminal vertices of $a\in A$ are described by $o(a), t(a)\in V$, respectively. Note that $o(\bar{a})=t(a)$ and $t(\bar{a})=o(a)$. 
The support edge of $a\in A$ is the undirected edge and denoted by $|a|$, which is omitted the direction, so that $|a|=|\bar{a}|$. We set $E=\{|a| \;|\; a\in A\}$ which is called the edge set. 
The degree of $v\in V$ is defined by $ \deg(v)=\# \{ a\in A  \;|\; t(a)=v\}$. 
We assume that the degree is uniformly bounded, that is, there exists a positive constant value $c$ such that $0<\sup_{u\in V} \deg(u)<c$, but we treat both finite and infinite $\# V$ cases. 
\subsection{Discrete-time $\epsilon$-Szegedy walk}
The total Hilbert space of the quantum walk on graph treated here is 
\[ \mathcal{A}=\ell^2(A)=\left\{ \psi: A\to \mathbb{C} \; : \; \sum_{a\in A}|\psi(a)|^2<\infty \right\} \]
whose inner product is standard. 
Let $\mathcal{H}_1$, $\mathcal{H}_2$ be Hilbert spaces and let $\Theta: \mathcal{H}_1\to \mathcal{H}_2$ be a linear operator. In this paper, the adojoint of $\Theta$ is defined by 
$\Theta^*: \mathcal{H}_2\to \mathcal{H}_1$ such that 
\[ \langle \psi,\Theta\phi \rangle_{\mathcal{H}_2} =\langle \Theta^*\psi,\phi \rangle_{\mathcal{H}_1} \]
for any $\phi\in \mathcal{H}_1$ and $\psi\in \mathcal{H}_2$. 
Set $X_u:=\{ a\in A\;|\; t(a)=u \}$ ($u\in V$) and $Y_e:=\{a\in A \;|\; |a|=e \}$ ($e\in E$). 
Note that $\# X_u=\deg(u)<\infty$ and $\# Y_e=2$. 
The symmetric arc set $A$ can be decomposed into the following two ways: 
\[ A=\bigsqcup_{u\in V} X_u = \bigsqcup_{e\in E}Y_e.  \]
For a countable set $\Omega$, we set $\mathbb{C}^{\Omega}$ as the vector space whose standard basis set is labeled by $\Omega$. 
Then the total Hilbert space is isomorphic to 
\begin{align} 
\mathcal{A} &\cong \left\{\psi\in \bigoplus_{u\in V}\mathbb{C}^{X_u} \;:\; \sum_{u\in V}||\psi(u)||_{\mathbb{C}^{X_u}}^2<\infty \right\} \label{eq:decom1} \\
&\cong \left\{ \varphi\in \bigoplus_{e\in E} \mathbb{C}^{Y_e}\;:\; \sum_{e\in E}||\varphi(e)||_{\mathbb{C}^{Y_e}}^2<\infty\right\}. \label{eq:decom2}  
\end{align}
The time evolution is defined by a unitary operator $U$ on $\mathcal{A}$, and it is given by product of two operators called shift operator $S$ and coin operator $C$ on $\mathcal{A}$, that is, $U=SC$. 
Set a local unitary operator on $\mathbb{C}^{X_u}$ for each $u\in V$ by $C_u$. 
The coin operator $C$ is denoted by 
\[C=\bigoplus_{u\in V} C_u \]
under the decomposition of (\ref{eq:decom1}). 
In the same way, setting a local unitary operator on $\mathbb{C}^{Y_e}$ for each $e\in E$ by $S_e$, we define the shift operator as 
\[ S=\bigoplus_{e\in E} S_e \]
under the decomposition of (\ref{eq:decom2}). 
Throughout this paper, we set the local unitary operators $C_u$ by an involution matrix; that is, $C_u^2=I_{mathbb{C}^{X_u}}$. Note that such an involution matrix can be expressed by   
\[ C_u=2\Pi_u-I_{\mathbb{C}^{X_u}}, \] 
where $\Pi_u$ is a projection onto a subspace in $\mathbb{C}^{X_u}$, this is an extension of the Grover matrix. Indeed, if we set $\Pi_u:=(1/\deg (u))J_u$, then $C_u$ becomes the Grover matrix ${\rm Gr}(\deg (u))$, where $J_u$ is the all $1$ matrix on $\mathbb{C}^{X_u}$. 
Let $\mathcal{V}$ be a Hilbert space so that the CONS is isomorphic to that of 
$\bigoplus_{u\in V}\left(\Pi_u\mathbb{C}^{X_u}\right)$. 
The boundary operator $d: \mathcal{A}\to\mathcal{V}$ is a map so that $d^*d=\bigoplus_u \Pi_u $. Thus $C$ is described by 
\[ C=2d^*d-I_{\mathcal{A}}. \]
See Section \ref{sect:generalcases} for more detail. 

On the other hand, by setting $0<\epsilon<1$, we focus on the shift operator as 
\[ S_e=\sqrt{1-\epsilon^2} \sigma_0+ i\epsilon \sigma_1, \]
where $\sigma_0$ is the identity matrix and $\sigma_1$ is the Pauri matrix.
We can check such a shift operator $S$ becomes unitary, and set this by $S(\epsilon)$ and the total time evolution operator by $U(\epsilon)=S(\epsilon)C$.
Then, we see that
\[
    (S(\epsilon)\psi)(a) = \sqrt{1-\epsilon^2}\psi(a) + i\epsilon\psi(\bar{a})
\]
for any $a\in A$ and $\psi\in \mathcal{A}$. 
Let us set a unitary self-adjoint operator $S_o$ on $\mathcal{A}$ by the flip flop shift such that 
\[ (S_o\psi)(a) = \psi(\bar{a}) \]
for any $a\in A$ and $\psi\in \mathcal{A}$. 
Here $S(\epsilon) = \sqrt{1-\epsilon^2}I_{\mathcal{A}} + i\epsilon S_o$ holds.
Hereafter, we abbreviate $S(\epsilon)$ as $S$.

Let us explain our motivation introducing the parameter $\epsilon$. 
If we set $\epsilon\to 1$, then the usual shift operator called the flip flop shift is reproduced and quantum walker moves to a neighbor vertex at every time step. On the other hand, if we set $\epsilon\to 0$, this quantum walk cannnot move to a neighbor vertex; just stay at the same place. Thus by setting this parameter $\epsilon$ between $[0,1]$, we can control the strength of the ``mobility" of the quantum walk. So we call the parameter $\epsilon$ the mobility parameter. 
In particular, we are interested in the behavior for the case that the moving is ``rare"; that is, $\epsilon\ll 1$ to connect a continuous-time quantum walk. This is an analogy of the Poisson's bell of the moving in the continuous-time random walk.   

The time iteration of our quantum walk model is described by 
\begin{equation}\label{eq:DTQW}
    \psi_{n+1}=U^2(\epsilon/2)\psi_n
\end{equation} 
with some initial state $\psi_0\in \mathcal{A}$ for  $n=0,1,2,\dots$. 
In this paper, the walk whose time evolution operator is constructed by two unitary involution operators is called the Szegedy walk.  
In particular, We call the quantum walk in (\ref{eq:DTQW}) the discrete-time $\epsilon$-Szegedy walk.

Let us the final time to quit the walk be $N$. 
The reason for the square of $U(\epsilon/2)$ will be seen in Section~\ref{sect:connectionCD}. 
To connect a corresponding continuous-time quantum walk, we will set $\epsilon=t/N$ with some constant positive real value $t$ and take the limit by
$\lim_{N\to\infty} \psi_N$. 
Note that until $n\leq N-1$, the time iteration follows (\ref{eq:DTQW}) with $\epsilon=t/N$, that is, 
\[ \psi_{n+1}=U^2(\tfrac{t}{2N}) \; \psi_n \text{ for $n=0,1,\dots,N-1$}. \]
\subsection{The discriminant operator}
\subsubsection{Grover walk on finite graph case}
First, we consider the simple case where $C_v=\mathrm{Gr}(\deg(v))$ for any $v\in V$ and the underlying graph is finite, so, the Hilbert space $\mathcal{A}$ is identified with $\mathbb{C}^{A}$, then operators $U(\epsilon), C, S(\epsilon)$ are also identified with matrices of $\mathbb{C}^{A\times A}$.
Note that because of the definition of the Grover matrix, $\Pi_u$ is the orthogonal projection onto the unit vector $(1/\sqrt{\deg(u)})[1,\dots,1]^\top$. 
Let $B\in \mathbb{C}^{V \times A}$ (which will be extended to a boundary operator $d$ in the next subsection) be the matrix satisfying 
\[ B(u,a)=
\begin{cases}
\frac{1}{\sqrt{\deg(u)}} & \text{: $t(a)=u$}\\
0 & \text{: otherwise}
\end{cases} \]
It holds that 
\[ BB^*=I_{\mathcal{A}} \text{ and }C=2B^*B-I_{\mathcal{A}}, \]
when $C_u$ is the Grover matrix for every $u\in V$. 
Let us set $T:=BS_oB^*$, which can be regarded as the matrix $\mathbb{C}^{V\times V}$; that is, 
\begin{align*}
T(u,v)=\begin{cases} \frac{1}{\sqrt{\deg(u)\deg(v)}} & \text{: $u$ and $v$ are connected in $G$}\\ 0 & \text{: otherwise.} \end{cases}
\end{align*}
Note that the symmetric matrix $T$ is isomorphic to the transition matrix of the simple random walk $P=D^{-1/2}TD^{1/2}$, where $D$ is the degree matrix such that
\[ D(u,v)=\begin{cases} \deg(u) & \text{: $u=v$}\\ 0 & \text{: otherwise.}\end{cases} \]
The symmetric matrix $T$ will play an important role to connect our discrete-time quantum walk and a continuous-time quantum walk. We call $T$ the discriminant operator. 

\subsubsection{General case}\label{sect:generalcases}
Secondly, let us extend the discriminant operator in our general setting; that is, $0<p_u:=\dim(\ker(I_{\mathbb{C}^{X_u}}-C_u))<\deg(u)$. Note that in the previous Grover walk case, $\ker(I_{\mathbb{C}^{X_u}}-C_u)$ is generated by the uniform vector, so $p_u=1$ for any $u\in V$. 
Let $\tilde{V}$ be the set $\tilde{V}:=\{ (u;j)\;|\; u\in V, j\in \{1,\dots,p_u\} \}$. 
The Hilbert space $\ell^2(\tilde{V})$ is denoted by $\mathcal{V}$. 
For any $f\in \mathcal{V}$ and $u\in V$, we define 
\[f[u]:=[f(u;1),\dots,f(u;p_u)]^\top \in \mathbb{C}^{\{1,\dots,p_u\}}. \]
Set $\{\xi_{u}^{(j)}\}_{j=1}^{p_u}$ as a CONS of $\ker(I_{\mathbb{C}^{X_u}}-C_u)\subset \mathbb{C}^{X_u}\cong \mathbb{C}^{\{1,\dots,\deg(u)\}}$. 
For $u\in V$ with $X_u=\{a_1,\dots,a_{\deg(u)}\}$, 
let a $p_u\times \deg(u)$-matrix $K_u$ be 
\[ K_u:=[\; \xi_u^{(1)}\;|\; \cdots \;|\; \xi_u^{(p_u)} \;]^*=[\; w_{a_1}\;|\;\cdots\;|\; w_{a_{\deg(u)}} \;]. \]
Note that for any $a\in A$,  the vector $w_a$ belongs to $\mathbb{C}^{\{1,\dots,p_{t(a)}\}}$ and can be expressed by
\begin{equation}
\label{eq:wa}
w_a=[\;\xi_{t(a)}^{(1)}(a),\cdots,\xi_{t(a)}^{(p_{t(a)})}(a) \;]^*.  
\end{equation}
Let $\iota_u: \mathcal{A}\to \mathbb{C}^{X_u}$ such that 
\[ (\iota_u\psi)(a)=\psi(a) \]
for any $a\in X_u$. 
The adjoint of $\iota_u$ is described by 
\[ (\iota_u^*\phi)(a)=\begin{cases} \phi(a) & \text{: $a\in X_u$}\\ 0 & \text{: otherwise.} \end{cases} \]
Let $d: \mathcal{A}\to \mathcal{V}$ be the map denoted by 
\begin{equation*} 
(d\psi)(u;j)=\langle\; \xi_{u}^{(j)}, \;\iota_u\psi \;\rangle_{\mathbb{C}^{X_u}} 
\end{equation*}
for any $u\in V$ and $j=1,\dots,p_u$, which is equivalent to 
\begin{equation}\label{eq:d}
    (d\psi)[u]=K_u\;\iota_u\psi 
\end{equation} 
Let us check that its adjoint is described by 
\begin{equation}\label{eq:adjoint}
(d^*f)\;(a)= \langle\; w_a,\;f[t(a)]\; \rangle_{\mathbb{C}^{\{1,\dots,p_u\}}} \end{equation}
for any $f\in \mathcal{V}$ and $a\in A$ as follows:
\begin{align*}
    \langle  d\psi, f \rangle &= \sum_{(u;j)} f(u;j) \overline{\langle\; \xi_{u}^{(j)}, \;\iota_u\psi \;\rangle} \\
    &= \sum_{(u;j)} f(u;j)\sum_{a\in A} (\iota_u^*\xi_u^{(j)})\;(a)\;\overline{\psi(a)} \\
    &= \sum_{a\in A} \left(\sum_{j=1}^{p_{t(a)}} f(t(a);j)\; (\iota_{t(a)}^*\;\xi_{t(a)}^{(j)})\;(a)\right)\; \overline{\psi(a)} \\
    &= \sum_{a\in A} \;\langle\; w_a,\;f[t(a)]\; \rangle \; \overline{\psi(a)}. 
\end{align*}
By the expression of (\ref{eq:adjoint}), it holds that  
\begin{equation}\label{eq:d^*}
\iota_ud^*f=K^*_u\; f[u] 
\end{equation}
for any $f\in \mathcal{V}$ and $u\in V$. 
We have the following lemma
\begin{lemma}
\label{lem:coisometry}
Let $d$ and $C$ be defined as the above. Then we have 
\begin{align*}
    dd^*&= I_{\mathcal{V}}, \\
    C &= 2d^*d-I_{\mathcal{A}}.   
\end{align*}
\end{lemma}
\begin{proof}
By (\ref{eq:d}) and (\ref{eq:d^*}), we have 
\begin{align*}
    (dd^*f)[u] &= K_u\iota_u\; d^*f= K_uK_u^*\;f[u] \\
    &= f[u].
\end{align*}
Here we used $\langle  \xi_u^{(i)},\xi_u^{(j)} \rangle=\delta_{i,j}$ since $\{\xi_u^{(j)}\}_{j=1}^{p_u}$ is a CONS of $\ker(I_{\mathbb{C}^{X_u}}-C_u)$.
By (\ref{eq:d}), (\ref{eq:d^*}), we have 
\begin{align*}
    d^*d\psi&= \sum_{u\in V}\iota_u^*\iota_u (d^*d)\psi 
    = \sum_{u\in V}\iota_u^*K_u^*(d\psi)[u] 
    = \sum_{u\in V}\iota_u^* K_u^*K_u\iota_u \psi \\
    &= \sum_{u\in V}\iota_u^* \Pi_u \iota_u \psi
\end{align*}
for any $\psi\in \mathcal{A}$. 
Since $C_u=2\Pi_u-I_{\mathbb{C}^{X_u}}$, we have 
\[ C=\sum_{u\in V} \iota_u^*(2\Pi_u-I_u)\iota_u =2d^*d-I_{\mathcal{A}}. \]
\end{proof}
\begin{definition}
The discriminant operator on $\mathcal{V}$ is denoted by  
\[T=dS_od^*.  \]
We remark that this $T$ becomes a self-adjoint operator, since $S_o = S_o^*$.
\end{definition}
Let $\tilde{w}_{a}\in \mathcal{V}$ be the extension of $w_a$ such that 
\[ (\tilde{w}_a)[u]=\begin{cases} w_a & \text{: $t(a)=u$,}\\ 0 & \text{: otherwise.} \end{cases} \]
Since $d\cong \bigoplus_{u\in V}\iota_u^* K_u\iota_u$,
we have 
\[ T=  \sum_{a\in A }\tilde{w}_{\bar{a}}\;\tilde{w}_{a}^*.  \]
Let $P_u$ be the projection onto $\mathrm{span}\{\delta_{(u;j)} \;|\; j=1,\dots,p_u\}\in \mathcal{V}$. 
Then we have the matrix valued entry of $T$ for $u,v\in V$ as
\begin{equation}\label{eq:Telement}
P_v T P_u= \sum_{t(a)=v,\;o(a)=u} w_{\bar{a}}\;w_{a}^*.  
\end{equation}
This means that the weight associated to moving a walker from $u$ and $v$ is the matrix represented by $w_{\bar{a}}\;w_{a}^*$. 

\subsubsection{Examples}
In the following, let us give some examples other than the Grover walk. 
\begin{enumerate}
\item Example of $p_u=1$ case (for any $u\in V$): Let us reproduce the discrete-time quantum walk introduced in \cite{Childs2010} which is induced by an arbitrary Hamiltonian operator on $\ell^2(V)$. 
Assume the underlying graph is finite and connected. 
Let $H$ be a Hamiltonian operator of the graph, and $\mathrm{abs}(H)$ be the elementwise absolute value of $H$. Set $\lambda_{max}$ as the maximal eigenvalue of $\mathrm{abs}(H)$ and 
$\nu_{max}$ as its eigenvector. 
The unit vector assigned at vertex $u\in V$,  $\xi_u:=\xi_u^{(1)}$, is denoted by 
\[ \xi_u(a)=\frac{1}{\sqrt{\lambda_{max}}}\sqrt{\frac{(H^*)_{u,o(a)}\;\nu_{max}(o(a))}{\nu_{max}(u)}} \]
for any $a\in A$  with $t(a)=u$. 
Then it is easy to see that the discriminant operator $T$ is described by 
\[ (T)_{u,v}=(dS_od^*)_{u,v}=(H)_{u,v}/\lambda_{max}.  \]
Note that our induced discrete-time quantum walk with the parameter $\epsilon$, namely the discrete-time $\epsilon$-Szegedy walk, can approximate the continuous-time quantum walk driven by $H/||\mathrm{abs}(H)||$.  
On the other hand, the discrete-time quantum walk in \cite{Childs2010} to approximates the continuous-time quantum walk is not the same as our discrete-time quantum walk model, which is induced by the lazy random walk with the  transition probability from $o(a)$ to $t(a)$ for any $a\in A\cup V$ as follows: 
\[ \begin{cases}
\epsilon \;\xi_{o(a)}^2(a) & \text{: $a\in A$,}\\ 1-\epsilon & \text{: $a\in V$.} \end{cases}\]
Thus the total Hilbert space in \cite{Childs2010} should be enlarged as $\ell^2(A\cup V)$ because the self-loop is added to every vertex by the underlying lazy random walk. 

\item Example of $p_u=2$ case (for any $u\in V$):  The underlying graph is set as the $3$-dimensional lattice. 
The arc whose terminal vertex is $\bs{x}=(x_1,x_2,x_3)\in \mathbb{Z}^3$ and origin vertex is $\boldsymbol{x}\mp \boldsymbol{e}_j$ is denoted by $(\boldsymbol{x};\pm j)$ $(j=1,2,3)$. 
Here $\bs{e}_1=(1,0,0)$, $\bs{e}_2=(0,1,0)$ and $\bs{e}_3=(0,0,1)$. 
Set $X_{\bs{x}}:=\{ (\bs{x};j)\;|\;j=\pm 1, \pm 2, \pm 3 \}$. 
For any $\bs{x}\in \mathbb{Z}^3$, 
the standard basis set of $\mathbb{C}^{X_{\bs{x}}}$ is denoted by 
\[ \{ \delta_{(\bs{x};+1)}, \delta_{(\bs{x};-1)},\delta_{(\bs{x};+2)}, \delta_{(\bs{x};-2)},\delta_{(\bs{x};+3)},\delta_{(\bs{x};-3)}\}. \]
Let us set $\xi_{\bs{x}}^{(1)}$,  $\xi_{\bs{x}}^{(2)}\in \mathbb{C}^{X_{\bs{x}}}$ by 
\[ \xi_{\bs{x}}^{(1)}=\frac{1}{\sqrt{6}}[1\;1\;\omega\;\omega\;\omega^2\;\omega^2]^\top,\; \xi_{\bs{x}}^{(1)}=\frac{1}{\sqrt{6}}[1\;1\;\omega^2\;\omega^2\;\omega\;\omega]^\top, \]
where $\omega=e^{2\pi i/3}$. 
The resulting quantum coin at $\bs{x}\in \mathbb{Z}$ is described by
\[ C_{\bs{x}}=2 (\xi_{\bs{x}}^{(1)}{\xi_{\bs{x}}^{(1)}}^*+\xi_{\bs{x}}^{(2)}{\xi_{\bs{x}}^{(2)}}^*)-I_{\mathbb{C}^{X_{\bs{x}}}}. \] 
Let us put $\sigma: \mathbb{C}^{X_{\boldsymbol{x}}}\to \mathbb{C}^{X_{\boldsymbol{x}}}$ as  the permutation matrix of the transposition $(\boldsymbol{x};j)\mapsto (\boldsymbol{x};-j)$ $(j=\pm 1,\pm 2,\pm 3)$ and $\mathrm{Gr}(k)$ is the $k$-dimensional Grover matrix, that is, $\mathrm{Gr}(k)=(2/k)\; J_k-I_k$, where $J_k$ is the all $1$ matrix. 
Then the coin matrix is equivalent to 
\[C_{\bs{x}}=-\sigma \mathrm{Gr}(6).\]
 
Let us see that this quantum walk driven by $C_{\bs{x}}$ is essentially same as the Grover walk with the {\it moving shift} on $\mathbb{Z}^3$ in the following. 
The moving shift operator $S_m$ is defined by 
$(S_m\psi)(\bs{x};j)=\psi(\bs{x}-\bs{e}_j;j)$ for any $j\in\{\pm 1,\pm 2,\pm 3\}$ and $\bs{x}\in \mathbb{Z}^3$. On the other hand, $(S_o\psi)(\bs{x};j)=\psi(\bs{x}-j;-j)$. 
It is easy to see that $(S_oS_m\psi)(\bs{x};j)=\psi(\bs{x};-j)$, which is local. 
Then the discrete-time quantum walk with the moving shift operator is expressed by that of the flip flop shift type such that  
\[U_m:=S_mC=S_oS_oS_mC=S_oC',\] 
where $C'$ is the directsum of $\sigma \mathrm{Gr}(6)$ over all the vertices $\bs{x}\in \mathbb{Z}^3$.

By (\ref{eq:wa}), $w(a)$'s are computed by  
\[ w_{(\bs{x};\pm 1)}=\frac{1}{\sqrt{6}}\begin{bmatrix}1\\1\end{bmatrix},\; 
 w_{(\bs{x};\pm 2)}=\frac{1}{\sqrt{6}}\begin{bmatrix}\omega^2\\\omega\end{bmatrix},\; 
  w_{(\bs{x};\pm 3)}=\frac{1}{\sqrt{6}}\begin{bmatrix}\omega\\\omega^2\end{bmatrix}.  
 \]
Then by (\ref{eq:Telement}), 
the discriminant operator $T: \ell^2(\mathbb{Z};\mathbb{C}^2)\to \ell^2(\mathbb{Z};\mathbb{C}^2)$ is described by
\begin{multline*} (Tf)(\bs{x})=
W_1f(\bs{x}-\bs{e}_1)+W_{-1}f(\bs{x}+\bs{e}_1)\\
+W_2f(\bs{x}-\bs{e}_2)+W_{-2}f(\bs{x}+\bs{e}_2)\\
+W_3f(\bs{x}-\bs{e}_3)+W_{-3}f(\bs{x}+\bs{e}_3),
\end{multline*}
for any $f\in \ell^2(\mathbb{Z}^2;\mathbb{C}^2)$ and $\bs{x}\in \mathbb{Z}^3$, 
where
\[ W_1=W_{-1}=\frac{1}{6}\begin{bmatrix}1&1\\1&1\end{bmatrix},\;
W_2=W_{-2}=\frac{1}{6}\begin{bmatrix}1&\omega\\\omega^2 &1\end{bmatrix},\;
W_3=W_{-3}=\frac{1}{6}\begin{bmatrix}1&\omega^2\\\omega &1\end{bmatrix}.
\]
The matrices $W_{\pm 1}$, $W_{\pm 2}$ and $W_{\pm 3}$ are the weights associated with moving to $\pm \bs{e}_1$, $\pm \bs{e}_2$ and $\pm \bs{e}_3$, respectively.
\end{enumerate}
\subsection{Continuous-time Szegedy walk}
The following operator $H$ on $\mathcal{A}$ is a self adjoint because the operators $S_o$ and $C$ are selfadjoint.   
\begin{align}
    H:=\frac{1}{2}(S_o+CS_oC).
    \label{eq:hamiltonian}
\end{align}  
Then we define the time evolution of the continuous-time quantum walk on the underlying graph $G$ treated here is defined on $\mathcal{A}$ by   
\begin{equation}\label{eq:CTQW}
-i\frac{\partial}{\partial t}\phi_t = H\phi_t. 
\end{equation}
Note that we can also express it in $H = \frac{1}{2}C(U_o + U_o^*)=\frac{1}{2}(U_o + U_o^*)C$, where $U_o$ is a unitary operator on $\mathcal{A}$ defined by $U_o = S_oC$.
We call this continuous-time quantum walk as continuous-time Szegedy walk.
\begin{corollary}
\label{cor:radius_of_H}
For the above $H$, the following holds.
\[
    \|H\|\leq 1.
\]
\end{corollary}
\begin{proof}
It follows from the following calculation for any $\Psi\in\HC$.
\begin{align*}
    \|H\Psi\|^2 = \frac{1}{4}\inpr{C\Psi}{(U_o+U_o^*)^2C\Psi}
    \leq \frac{1}{4}(\|U_o\| + \|U_o^*\|)^2\|C\Psi\|^2 = \|\Psi\|^2.
\end{align*}
\end{proof}

Let $\mathcal{I}\subset \mathcal{A}$ be defined by $\mathcal{I}:=d^*\mathcal{V}+S_od^*\mathcal{V}\subset \mathcal{A}$ and called the inherited subspace.
Then, the orthogonal complement of this subspace $\mathcal{B}:=\mathcal{I}^{\perp}\subset\mathcal{A}$ is called the birth subspace.
If $\Psi\in\mathcal{B}$, then $\Psi$ satisfies the following equation for any $f_1,f_2\in\mathcal{V}$:
\[
    \inpr{\Psi}{d^*f_1 + S_od^*f_2} = \inpr{d\Psi}{f_1} + \inpr{dS_o\Psi}{f_2}=0.
\]
It gives $\mathcal{B} = \ker(d)\cap\ker(dS_o)$.
The following lemma guarantees that $\mathcal{I}$ and $\mathcal{B}$ will be invariant subspaces of $H$.
\begin{lemma}
\label{lem:invariant}
For $\mathcal{I}$ and $\mathcal{B}$ defined as above, these are also invariant subspace of $H$, that is both $H\mathcal{I}\subset \mathcal{I}$ and $H\mathcal{B}\subset \mathcal{B}$ hold.
\end{lemma}
\begin{proof}
Multiplying $d^*$ and $S_od^*$ from the right side for $H$, the following equations hold since $Cd^*=d^*$ and $S_o^2=I_{\mathcal{A}}$.
\begin{align}
    \label{eq:Hd^*}
    Hd^* &= \frac{S_od^* + CS_od^*}{2} = \frac{(I_{\mathcal{A}}+C)S_od^*}{2} = d^*T,
    \\
    \label{eq:HSd^*}
    HS_od^* &= \frac{d^* + CS_oCS_od^*}{2} 
    =\frac{d^*+(2d^*d-I_{\mathcal{A}})S_o(2d^*d-I_{\mathcal{A}})S_od^*}{2} 
    =2d^*T^2-S_od^*T.
\end{align}
For any $\Psi\in\mathcal{I}$ written by $\Psi = d^*f_1 + S_od^*f_2$ with $f_1,f_2\in\mathcal{V}$, above equations give that
\[
    H(d^*f_1+S_od^*f_2)
    =d^*(Tf_1+2T^2f_2)-S_od^*Tf_2\in\mathcal{I}.
\]
Thus, $H\mathcal{I}\subset \mathcal{I}$ holds.
Next, $H\mathcal{B}\subset \mathcal{B}$ can be shown immediately by $\mathcal{B}=\ker(d)\cap\ker(dS_o)$ and the following equations:
\begin{align}
\label{eq:dH}
dH &=\frac{1}{2} d(CS_o+S_oC) = \frac{1}{2}(dS_o+dS_o(2d^*d-I_{\mathcal{A}}))= Td,
\\
\nonumber
dS_oH &= \frac{1}{2}d(S_oCS_o + C) = \frac{1}{2}(dS_o(2d^*d-I_{\mathcal{A}})S_o + d)=TdS_o.
\end{align}
\end{proof}

This continuous-time Szegedy walk on the Hilbert space generated by $A$ can reproduce a typical continuous-time quantum walk (discrete-Schr{\"o}dinger equation) on $\mathcal{V}$ driven by the  Hamiltonian $T$ as follows.  
\begin{proposition}\label{prop:DiscSch}
Let us consider the following discrete-Schr{\"o}dinger equation on $\mathcal{V}$: 
\[ -i\frac{\partial}{\partial t}f_t=T f_t,\;f_0=g. \]
Then the solution $f_t$ is equivalent to $d\phi_t$, where $\phi_t$ is the solution of the following Schor{\"o}dinger equation on $\mathcal{A}$:  
\[ -i\frac{\partial}{\partial t}\phi_t=H \phi_t,\;\phi_0=d^*g. \]
\end{proposition}
This means that the Schor{\"o}dinger equation on $\mathcal{A}$ driven by $H$ can reproduce that on $\mathcal{V}$ driven by $T$. 
\begin{proof}
From \eqref{eq:dH}, we have $Hd^* = d^*T$.
Thus $\phi_t\in \mathcal{A}$ is expressed by some $f_t\in \mathcal{V}$ such that $\phi_t=d^*f_t$, which is equivalent to $d\phi_t=f_t$ since $dd^*=I_{\mathcal{V}}$.   
Then we have  
\begin{align*}
-i\frac{\partial}{\partial t}f_t=
-i\frac{\partial}{\partial t} d\phi_t &= dH d^* f_t=Tf_t. 
\end{align*}
This completes the proof. 
\end{proof}



We note that, subspaces $\mathcal{I}$ and $\mathcal{B}$ are also invariant subspaces of our discrete-time quantum walk, that is both $U(\epsilon)\mathcal{I}\subset \mathcal{I}$ and $U(\epsilon)\mathcal{B}\subset \mathcal{B}$ hold.
This property plays very important role for spectral analysis of discrete-time quantum walks.

As we will see in Section~\ref{sect:CTQW_Spec}, the continuous-time quantum walk has the same eigensapce restricted to $\mathcal{B}$ as that of discrete-time quantum walk $U_o=S_oC$. 
Then, for example, if the local quantum coin $C_u$ is given by the Grover matrix, then the eigenspace restricted to $\mathcal{B}$ is generated by the cycle and path information~\cite{HPSS}.
This eigenstate works as the dark state. For examples, once a closed cycle exists in the underlying graph, so called the localization also occurs in this continuous-time quantum walk if the graph is infinite~\cite{HKSS}, while the non-zero survival probability of this continuous-time quantum walk on a finite graph with sink can be observed~\cite{MNSJ}. 

\section{Connecting the discrete and continuous-time Szegedy walks: proof of Theorem~\ref{thm:main}}\label{sect:connectionCD}
By setting a sufficiently small parameter $\epsilon$, the amplitude of our discrete-time quantum walk to the neighbors is quite small. 
Such a dynamics seems to be like a continuous-time random walk whose moving follows Poisson's bell. Indeed, we obtain the following proposition which gives the connection from our discrete-time quantum walk to the corresponding  continuous-time quantum walk. 
\begin{proposition}
\label{thm:DTtoCT1}
Let us change the time interval of the discrete-time quantum walk (\ref{eq:DTQW}) by $\epsilon$, such that $\psi_{\tau+\epsilon}=U(\epsilon/2)^2\psi_\tau$ for $\tau\in\{0,\epsilon,2\epsilon,\dots\}$. 
Then the discrete-time quantum walk (\ref{eq:DTQW}) is a difference approximation of the continuous-time quantum walk (\ref{eq:CTQW}) in the following meaning:  
\[ -i\frac{\psi_{\tau+\epsilon}-\psi_\tau}{\epsilon}= H\psi_\tau+O(\epsilon) \]
for sufficiently small $\epsilon$. 
\end{proposition}
\begin{proof}
The shift operator $S(\epsilon/2)$ can be written by
\[ S(\epsilon/2)=\sqrt{1-(\epsilon/2)^2}I_{\mathcal{A}} + i\frac{\epsilon}{2}S_o = I_{\mathcal{A}} + i\frac{\epsilon}{2} S_o+O(\epsilon^2)\] 
for small $\epsilon \ll 1$. 
Noting $C^2=I$, we have 
\begin{align*} 
U^2(\epsilon/2) &= I_{\mathcal{A}} + i \frac{\epsilon}{2}(S_o + CS_oC) +O(\epsilon^2). 
\end{align*}
Since $\psi_{\tau+\epsilon} = U^2(\epsilon /2)\psi_\tau$, we have 
\[ -i\frac{\psi_{\tau+\epsilon}-\psi_\tau}{\epsilon}= H\psi_\tau+O(\epsilon). \]
\end{proof}
This proposition suggests that our discrete-time quantum walk on $G$ can approximate the continuous-time quantum walk on the same graph. To see the way to completely identify with the continuous-time quantum walk, let us show again our main theorem displayed in Section~1: 
\begin{reptheorem}{thm:main}
Let $G=(V,A)$ be the underlying graph which is simple and connected. 
Let $\psi_{N}^{[D,\epsilon]}$ be the state at the final time $N$ of the discrete-time quantum walk  (\ref{eq:DTQW}) with the time evolution operator $U(\epsilon/2)^2$ and the initial state $\varphi_0\in \mathcal{A}$, that is, 
\[ \psi_0^{[D,\epsilon]}=\varphi_0;\;\;\psi_{n}^{[D,\epsilon]}=U^2(\epsilon/2)\;\psi_{n-1}^{[D,\epsilon]} \;\text{ $(n=1,\dots,N)$}, \]
while $\phi_{t}^{[C]}$ be the $t$-th iternation of the continuous-time quantum walk (\ref{eq:CTQW}) with the Hermitian operator $H$ and the same initial state $\varphi_0$, that is,
\[ \phi_0^{[C]}=\varphi_0;\;\; -i\frac{\partial}{\partial t}\phi_t^{[C]} =H\phi_t^{[C]} \;\text{ $(t>0)$}. \]
Then we have 
\begin{equation*}
\phi_t^{[C]}(a) = \lim_{N\to\infty} \psi^{[D,\;t/N]}_N(a)
\end{equation*}
for any $a\in A$ and $t>0$.
\end{reptheorem}
Now let us give the proof of Theorem~\ref{thm:main}. 
\begin{proof}
Let a linear operator $L:\mathcal{V}^2\to\mathcal{A}$ by 
\begin{align}
        \label{eq:L}
    L=[d^*,\ S_od^*], 
\end{align}
and an operator matrix $\tilde T$ on $\mathcal{V}^2$ as follows:
\begin{align}
\label{eq:tilT}
\tilde{T}=\begin{bmatrix}
T & 2T^2 \\ O & -T
\end{bmatrix}.
\end{align}
Note that the range of operator $L$ equals to $\mathcal{I}$.
Here, equations \eqref{eq:Hd^*} \eqref{eq:HSd^*} give the following key relation:
\begin{align}
\label{eq:HL=LtilT}
    HL = L\tilde{T}.
\end{align}
It says that the continuous-time quantum walk with initial state $\varphi_I=Lf_I\in\mathcal{I},\ f_I\in\mathcal{V}^2$ is written as follows:
\begin{align*}
    e^{itH}\varphi_I = L\, \exp{it \tilde{T}}f_I.
\end{align*}
Moreover, for $\varphi_B\in\mathcal{B}$, we can check that
\[
    H\varphi_B = \frac{1}{2}(S_o+CS_o(2d^*d-I_{\mathcal{A}}))\varphi_B
    =
    \frac{1}{2}(I_{\mathcal{A}} - C)S_o\varphi_B
    =
    (I_{\mathcal{A}} - d^*d)S_o\Psi_B = S_o\varphi_B,
\]
so the following holds.
\begin{align*}
    e^{itH}\varphi_B = \exp{itS_o}\varphi_B.
\end{align*}
Therefore, $t$-th iteration of the continuous-time quantum walk $\phi_t^{[C]}$ with initial state $\varphi_0 = \varphi_I + \varphi_B$ is given as follows.
\begin{align}
\label{eq:MainThmCT}
    \phi_t^{[C]} = L\, \exp{it \tilde{T}}f_I + \exp{itS_o}\varphi_B.
\end{align}
Next, we consider the discrete-time quantum walk.
Put $p=\epsilon/2$ and $q=\sqrt{1-(\epsilon/2)^2}$.
Then, we have
\begin{align*}
    U(p)d^* &= (qI_{\mathcal{A}} + ipS_o)(2d^*d - I_{\mathcal{A}})d^* = qd^* + ipS_od^*,
    \\
    U(p)S_od^* &= (qI_{\mathcal{A}} + ipS_o)(2d^*d - I_{\mathcal{A}})S_od^* = d^*(2qT-ip)+S_od^*(-q+2ipT).
\end{align*}
Thus two-times iteration of $U(p)$ is written as follows.
\begin{align*}
    U(p)^2 d^* &= (1 + 2ipq T)d^* -2p^2TS_od^*,
    \\
    U(p)^2S_od^* &= (4ipqT^2 +2p^2T)d^* + (1 -2ipqT)S_od^*.
\end{align*}
It says that the time evolution of discrete-time quantum walk with initial state $\varphi_I$ is calculated by
$U(p)^2\varphi_I =L\tilde T_I(p)f_I$, where
\[
    \tilde T_I(p) = 
    I_{\mathcal{V}^2} + 2ipq\begin{bmatrix}
        T & 2T^2 \\ O & -T
    \end{bmatrix}
    -2p^2\begin{bmatrix}
        O & -T \\ T & O
    \end{bmatrix}
    =
    I_{\mathcal{V}^2} + 2ipq \tilde T
    -2p^2\begin{bmatrix}
        O & -T \\ T & O
    \end{bmatrix}.
\]
Similarly, we have
\[
    U(p)\varphi_B = (qI_{\mathcal{A}} + ipS_o)(2d^*d - I_{\mathcal{A}})\varphi_B
    =
    -(qI_{\mathcal{A}} + ipS_o)\varphi_B,
\]
and $U^2(p)\varphi_B =\tilde T_B(p)\varphi_B$
with
\[
    \tilde T_B(p) = (qI_{\mathcal{A}} + ipS_o)^2 = I_{\mathcal{A}} + 2ipq S_o -2p^2I_{\mathcal{A}}.
\]
Therefore the state at the final time $N$ of the dicrete-time quantum walk with parameter $\epsilon = 2p =t/N$ and initial state $\varphi_0$ is given as follows.
\begin{align}
\label{eq:MainThmDT}
    \psi_N^{[D, t/N]} =U(\tfrac{t}{2N})^{2N}(\varphi_I + \varphi_B) = L\, \tilde T_I\left(\tfrac{t}{2N}\right)^N f_I + \tilde T_B\left(\tfrac{t}{2N}\right)^N\varphi_B.
\end{align}
Since $\mathcal{I}$ and $\mathcal{B}$ are both invariant subspace of $H$ and $U(\epsilon)$, equations \eqref{eq:MainThmCT} \eqref{eq:MainThmDT} say that it is sufficient to show the convergence of two operators.
Firstly, we show that $\lim_{N\to\infty}\tilde T_I\left(\tfrac{t}{2N}\right)^N = \exp{it\tilde T}$.
Note that the spectral radius of $\tilde T_I\left(\tfrac{t}{2N}\right)$ can almost be regarded as $1$ when $N$ is sufficiently large, so we have a Taylor expansion of $\log(\tilde T_I\left(\tfrac{t}{2N}\right))$ and get the following result with some boundary operators $\Theta_I$ and $\Theta_I'$.
\begin{align*}
    \lim_{N\to\infty}\tilde T_I\left(\tfrac{t}{2N}\right)^N &=
    \lim_{N\to\infty}
    \exp{N\log\left( I_{\mathcal{V}^2} + 2i\left(\tfrac{t}{2N}\right)\sqrt{1-\left(\tfrac{t}{2N}\right)^2} \tilde T
    -2
    \left(\tfrac{t}{2N}\right)^2
    \begin{bmatrix}
        O & -T \\ T & O
    \end{bmatrix}\right)}
    \\
    &= \lim_{N\to\infty}
    \exp{N\log\left( I_{\mathcal{V}^2} + 2i\left(\tfrac{t}{2N}\right)\tilde T + \left(\tfrac{t}{2N}\right)^2
    \Theta_{I}'
    \right)}
    \\
    &= \lim_{N\to\infty}
    \exp{N\left(i\tfrac{t}{N}\tilde T + \left(\tfrac{t}{2N}\right)^2
    \Theta_{I}
    \right)}
    \\
    &=
    \exp{it\tilde T}.
\end{align*}
By using the same technique, we can also show that $\lim_{N\to\infty}\tilde T_B\left(\tfrac{t}{2N}\right)^N = \exp{it S_o}$ as follows.
\begin{align*}
 \lim_{N\to\infty}\tilde T_B\left(\tfrac{t}{2N}\right)^N 
 &= \lim_{N\to\infty} \exp{N \log\left(I_{\mathcal{A}} + 2i\left(\tfrac{t}{2N}\right)\sqrt{1-\left(\tfrac{t}{2N}\right)^2}S_o -2 \left(\tfrac{t}{2N}\right)^2 I_{\mathcal{A}} \right)}
 \\
 &= \lim_{N\to\infty} \exp{N \log\left(I_{\mathcal{A}} + 2i\left(\tfrac{t}{2N}\right)S_o +
 \left(\tfrac{t}{2N}\right)^2
 \Theta_B'
 \right)}
 \\
 &= \lim_{N\to\infty} \exp{
N\left(
i\tfrac{t}{N}S_o + 
 \left(\tfrac{t}{2N}\right)^2
 \Theta_B
\right)
 }
 \\
 &=
 \exp{itS_o}.
\end{align*}
Thus the proof is completed.
\end{proof}
\begin{theorem}
\label{thm:order}
For any $t>0$, the following holds. 
\[
\|e^{itH} - U(\tfrac{t}{2N})^{2N}\| = \mathcal{O}(\tfrac{1}{N}).
\]
\end{theorem}
\begin{proof}
By the previous argument in the proof of Theorem \ref{thm:main}, it is sufficient to show that
\[
\left\|\exp{it\tilde T} - \tilde T_I(\tfrac{t}{2N})^{N}\right\| = 
\mathcal{O}(\tfrac{1}{N})
\quad \text{and}\quad 
\left\|\exp{itS_o} - \tilde T_B(\tfrac{t}{2N})^{N}\right\| = \mathcal{O}(\tfrac{1}{N}).
\]
Since the latter equality can be proved in exactly the same way as the former, here we prove only the former equality.
As shown in the previous discussion, we can write $\tilde T_I(\tfrac{t}{2N})^{N}$ by a bounded operator $\Theta$ on $\mathcal{V}^2$ such that
\begin{align*}
\tilde T_I(\tfrac{t}{2N})^{N} &=
\exp{N\left(i\tfrac{t}{N}\tilde T + 
\tfrac{1}{N^2}
\Theta
    \right)
    }
    =
    \exp{
    it \tilde T + \tfrac{1}{N}\Theta
    }.
\end{align*}
It is well known that $\|\exp{A+B} - \exp{A}\exp{B}\| = \mathcal{O}(\|AB\|)$ holds for any boundary operators $A$ and $B$, so we can complete the proof as follows.
\begin{align*}
\left\|\exp{it\tilde T} - \tilde T_I(\tfrac{t}{2N})^{N}\right\|
&=
\left\|
\exp{it\tilde T}
-
\exp{it\tilde T}
\exp{\tfrac{1}{N}\Theta}
+
\exp{it\tilde T}
\exp{\tfrac{1}{N}\Theta}
-
\tilde T_I(\tfrac{t}{2N})^{N}
\right\|
\\
&\leq 
\left\|
\exp{it\tilde T}
\left(
I_{\mathcal{V}^2}
-
\exp{\tfrac{1}{N}\Theta}
\right)
\right\|
+
\left\|
\exp{it\tilde T}
\exp{\tfrac{1}{N}\Theta}
-
\exp{it\tilde T + \tfrac{1}{N}\Theta}
\right\|
\\
&=\mathcal{O}(\tfrac{1}{N}).
\end{align*}
\end{proof}
From Theorem~\ref{thm:order}, we immediately obtain the following corollary which implies that 
the running time of the discrete-time quantum walk to approximate the corresponding continuous-time quantum walk with a small error $\delta$.
\begin{corollary}
For any fixed $t>0$, there exists a constant $c_0>0$, such that for any $\delta>0$, if $N>c_0/\delta$, then 
\[ || \phi_t^{[C]}-\psi_N^{[D,t/N]} ||<\delta. \] 
\end{corollary}

\section{Spectral analysis for continuous-time Szegedy walk}\label{sect:CTQW_Spec}
In this section, we show the spectrum of $H$, $\sigma(H)$ , defined in \eqref{eq:hamiltonian}.
We note that the identity operator will be omitted from this chapter, so we write $C=2d^*d-1$ for example.
In general, since $H$ is a self-adjoint operator, $\sigma(H)$ is decomposed by point spectrum $\sigma_p(H)$ and continuous spectrum $\sigma_c(H)$, that is $\sigma(H)=\sigma_p(H)\cup \sigma_c(H)$ holds.
The aim of this section is to prove the following theorem.
The method for analysing eigenvalues is referred to in \cite{MOS2017, KSY2022}, while the method for analysing continuous spectra is referred to in \cite{SS2019}.
\begin{theorem}
\label{thm:spectrum}
The spectrum of $H$ is given as follows.
\begin{align*}
    \sigma(H) &= \sigma(T)\cup\sigma(-T)\cup \sigma_p(H),
    \\
    \sigma_p(H) &= 
    \sigma_p(T)\cup\sigma_p(-T)
    \cup
    \{+1\}^{\dim \mathcal{B}_+}
    \cup
    \{-1\}^{\dim \mathcal{B}_-}
    ,
\end{align*}
where superscripts of sets denote the multiplicity of eigenvalues, and $\mathcal{B}_\pm = \ker(C+1)\cap\ker(S_o\mp 1)$.
That is, the multiplicity of eigenvalue $\lambda$ is
\begin{align*}
M_{H}(\lambda)
=
\begin{cases}
M_T(\lambda) + M_T(-\lambda),\quad & \lambda\neq \pm 1,
\\
M_T(\pm 1)+\dim\mathcal{B}_{\pm},\quad & \lambda =\pm 1,
\end{cases}
\end{align*}
where $M_X(\lambda)$ means the multiplicity of eigenvalue $\lambda$ of an operator $X$.
Furthermore, eigenspaces are induced by the following.
\begin{align*}
\ker(H-\lambda)
=
\begin{cases}
d^*\ker(T-\lambda)\oplus (\lambda +S_o)d^*\ker(T+\lambda),\quad & \lambda\neq \pm 1,
\\
S_od^*\ker(T-\lambda)\oplus\mathcal{B}_{\pm},\quad & \lambda =\pm 1.
\end{cases}
\end{align*}
\end{theorem}
\begin{proof}
Lemma \ref{lem:ev_inherited} and Lemma \ref{lem:ev_birth} give all eigenvalues(including these multiplicities) and eigenspaces of $H$.
In addition, Lemma \ref{lem:cont_lem3} gives the spectrum $\sigma(H)$.
\end{proof}

\subsection{Point spectrum}
\label{subsec:point_sp}
We now analyse to the point spectrum of the Hamiltonian $H$.
Since $\mathcal{A}$ is decomposed into $\mathcal{A} = \mathcal{I}\bigoplus \mathcal{B}$, it is sufficient to focus on the respective eigenspaces of $\ker(H-\lambda)\cap\mathcal{I}$ and $\ker(H-\lambda)\cap\mathcal{B}$.
The former is called ``inherited eigenspaces" and the latter is called ``birth eigenspaces".

\subsubsection{Inherited eigenspaces}
We now focus on the inherited eigenspace $\ker(H-\lambda)\cap\mathcal{I}$.
From \eqref{eq:HL=LtilT}, an element of $\mathcal{I}$ expressed by $\Psi=L{\bf f}\in\mathcal{I},\ {\bf f} = [f_1,\, f_2]^T\in\mathcal{V}^2$, then $\Psi\in \ker(H-\lambda)$ is equivalent to the following:
\[
    L(\tilde T -\lambda){\bf f}=0.
\]
This means that $\ker(H-\lambda)\cap\mathcal{I} = L \ker(L(\tilde T -\lambda))$.
\begin{lemma}
\label{lem:kerL_tilT}
For $L$ and $\tilde T$ defined in \eqref{eq:L} \eqref{eq:tilT}, the followings hold for $\lambda\in\mathbb{R}$.
\begin{align*}
    &(i)\quad \ker L=
    \Lr{
    \begin{bmatrix}
    1 \\ -1
    \end{bmatrix}
    \zeta_1
    +
    \begin{bmatrix}
    1 \\ 1
    \end{bmatrix}
    \zeta_{-1}
    \ \middle |\ 
    \zeta_{\pm 1}\in\ker(T\mp 1)
    },
    \\
    &(ii)\quad \ker(\tilde T-\lambda)=
    \Lr{
    \begin{bmatrix}
    1 \\ 0
    \end{bmatrix}
    \zeta_{\lambda}
    +
    \begin{bmatrix}
    \lambda \\ 1
    \end{bmatrix}
    \zeta_{-\lambda}
    \ \middle |\ 
    \zeta_{\pm \lambda}\in\ker(T\mp \lambda)
    }.
\end{align*}
\end{lemma}
\begin{proof}
At first, if ${\bf f}=[f_1,\, f_2]^T\in\ker L$, then $d^*f_1 + Sd^*f_2=0$ holds.
Multiplying this equation by $d$ and $dS$ respectively, we have the followings:
\begin{align*}
    f_1 + Tf_2 =0,\quad
    Tf_1 + f_2 = 0.
\end{align*}
By substituting one for the other, we have a necessary condition to ${\bf f}\in\ker L$ as $(1-T^2)f_1=0$.
Since $T$ is a self-adjoint operator, we write $f_1 = \zeta_{1} + \zeta_{-1}$ with some $\zeta_1\in\ker(T-1),\ \zeta_{-1}\in\ker(T+1)$.
By substituting this for the above equation again, we have $f_2=-\zeta_{1} + \zeta_{-1}$, and
\[
    {\bf f} = \begin{bmatrix}
    1 \\ -1
    \end{bmatrix}\zeta_{1}
    +
    \begin{bmatrix}
    1 \\ 1
    \end{bmatrix}\zeta_{-1}.
\]
The following calculations show the statement (i) of this lemma:
\begin{align*}
 \|L{\bf f}\|^2 &=
 \inpr{(1-S_o)d^*\zeta_{1}+(1+S_o)d^*\zeta_{-1}}{(1-S_o)d^*\zeta_{1}+(1+S_o)d^*\zeta_{-1}}
 \\
 &=\inpr{\zeta_{1}}{d(1-S_o)^2d^*\zeta_{1} + d(1-S_o^2)d^*\zeta_{-1}}+\inpr{\zeta_{-1}}{d(1-S_o^2)d^*\zeta_{1} + d(1+S_o)^2d^*\zeta_{-1}}
 \\
 &=2\inpr{\zeta_{1}}{(1-T)\zeta_{1}}+2\inpr{\zeta_{-1}}{(1+T)\zeta_{-1}}
 \\&=0.
\end{align*}
Next, we see that ${\bf f}=[f_1,\, f_2]^T\in\ker(\tilde T -\lambda)$ is equivalent to the following equations hold:
\begin{align*}
    (T-\lambda)f_1 + 2T^2f_2 &= 0
    \\
    -(T+\lambda)f_2 &= 0.
\end{align*}
These equations give $f_1 = \zeta_{\lambda} + \lambda\zeta_{-\lambda},\ f_2 = \zeta_{-\lambda},\ \zeta_{\pm\lambda}\in\ker(T\mp\lambda)$.
This shows the statement (ii) of the lemma.
\end{proof}
We now consider $\ker L(\tilde T -\lambda)$.
From Lemma \ref{lem:kerL_tilT} (i), $(\tilde T-\lambda){\bf f}\in\ker L,\ {\bf f} =[f_1,\, f_2]^{T}\in\mathcal{V}^2$ is equivalent to the following both equations hold with some $\zeta_{\pm 1}\in\ker(T\mp 1)$:
\begin{align}
    \label{eq:A}
    (T-\lambda)f_1 + 2T^2f_2 &= \zeta_{1} + \zeta_{-1},
    \\
    \label{eq:B}
    (T+\lambda)f_2 &= \zeta_{1} - \zeta_{-1}. 
\end{align}
\begin{lemma}
\label{lem:ev_inherited}
The following holds.
\begin{align*}
    \ker(H-\lambda)\cap\mathcal{I}=
    L\ker(L(\tilde T-\lambda))
    =
    \begin{cases}
    d^*\ker(T-\lambda) \bigoplus (\lambda + S_o)d^*\ker(T+\lambda)
    ,\quad &\lambda \neq \pm 1,
    \\
    S_od^*\ker(T\mp 1),\quad & \lambda = \pm 1.
    \end{cases}
\end{align*}
In particular, this result yields the following statement about the multiplicity of eigenvalues.
\begin{align*}
    \dim(\ker(H-\lambda)\cap\mathcal{I})=
    \begin{cases}
    M_T(\lambda) + M_T(-\lambda)
    ,\quad &\lambda \neq \pm 1,
    \\
    M_T(\pm 1),\quad & \lambda = \pm 1.
    \end{cases}
\end{align*}
\end{lemma}
\begin{proof}
$\ker(H-\lambda)\cap\mathcal{I}=L\ker(L(\tilde T-\lambda))$ has already been mentioned.
If $\lambda\neq\pm 1$ case, \eqref{eq:B} gives
\[
    f_2 = \zeta_{-\lambda} + \frac{1}{1+\lambda}\zeta_{1}- \frac{1}{-1+\lambda}\zeta_{-1},
\]
and substituting it for \eqref{eq:A} gives
\[
    f_1 = \zeta_{\lambda} + \lambda\zeta_{-\lambda} - \frac{1}{1+\lambda}\zeta_{1}- \frac{1}{-1+\lambda}\zeta_{-1},
\]
where $\zeta_{\pm 1}\in\ker(T\mp 1),\ \zeta_{\pm \lambda}\in\ker(T\mp \lambda)$.
Thus, Lemma \ref{lem:kerL_tilT} (i) and (ii) imply
\begin{align*}
    \ker(L(\tilde T -\lambda))
    &=
    \Lr{
    \begin{bmatrix}
    1 \\ -1
    \end{bmatrix}
    \zeta_{1}
    +
    \begin{bmatrix}
    1 \\ 1
    \end{bmatrix}
    \zeta_{-1}
    \ \middle | \ \zeta_{\pm 1}\in\ker(T\mp 1)
    }
    \oplus
    \Lr{
    \begin{bmatrix}
    1 \\ 0
    \end{bmatrix}
    \zeta_{\lambda}
    +
    \begin{bmatrix}
    \lambda \\ 1
    \end{bmatrix}
    \zeta_{-\lambda}
    \ \middle |\ 
    \zeta_{\pm \lambda}\ker(T\mp \lambda)
    }
    \\
    &=\ker L\oplus \ker(\tilde T -\lambda).
\end{align*}

If $\lambda = \pm 1$ case, \eqref{eq:B} becomes $(T\pm 1)f_2=\zeta_{1}-\zeta_{-1}$.
Here, we have $\zeta_{\mp 1}=0$ because
\[
    0 = \inpr{\zeta_{\mp 1}}{(T\pm 1)f_2}
    =\inpr{\zeta_{\mp 1}}{\zeta_{1}-\zeta_{-1}} = \mp \|\zeta_{\mp 1}\|^2.
\]
Thus, \eqref{eq:B} is $(T\pm 1)f_2= \pm \zeta_{\pm 1}$, and it gives $f_2 = \zeta'_{\mp 1} + \frac{1}{2}\zeta_{\pm1}$ with $\zeta'_{\mp 1}\in\ker(T\pm 1)$.
By instituting it for \eqref{eq:A}, we also have
\[
    (T\mp 1)f_1 + 2\zeta'_{\mp 1} + \zeta_{\pm 1} = \zeta_{\pm 1},
\]
so $f_1 = \zeta'_{\pm 1} \pm \zeta'_{\mp 1}$ with $\zeta'_{\pm 1}\in\ker(T\mp 1)$.
Thus, Lemma \ref{lem:kerL_tilT} (i) implies
\begin{align*}
    \ker(L(\tilde T-\lambda)) &=
    \Lr{
    \begin{bmatrix}
    0 \\ 1
    \end{bmatrix}
    \zeta_{\pm 1} + 
    \begin{bmatrix}
    1 \\ 0
    \end{bmatrix}
    \zeta'_{\pm 1} +
    \begin{bmatrix}
    1 \\ \pm 1
    \end{bmatrix}
    \zeta'_{\mp 1}
    \ \middle |\ 
    \zeta_{\pm 1},\zeta'_{\pm 1},\in\ker(T\mp 1),\, \zeta'_{\mp 1},\in\ker(T\pm 1)
    }
    \\
    &=
    \Lr{
    \begin{bmatrix}
    0 \\ 1
    \end{bmatrix}
    \zeta_{\pm 1} + 
    \begin{bmatrix}
    1 \\ \mp 1
    \end{bmatrix}
    \zeta'_{\pm 1} +
    \begin{bmatrix}
    1 \\ \pm 1
    \end{bmatrix}
    \zeta'_{\mp 1}
    \ \middle |\ 
    \zeta_{\pm 1},\zeta'_{\pm 1},\in\ker(T\mp 1),\, \zeta'_{\mp 1},\in\ker(T\pm 1)
    }
    \\
    &=
    \Lr{
    \begin{bmatrix}
    0 \\ 1
    \end{bmatrix}
    \zeta_{\pm 1}
    \ \middle |\ 
    \zeta_{\pm 1}\in\ker(T\mp 1)
    }
    +
    \ker L.
\end{align*}
In summary, the following is obtained.
\[
    L\ker(L(\tilde T -\lambda))
    =
    \begin{cases}
    L\ker(\tilde T -\lambda),\quad & \lambda\neq \pm 1,
    \\
    L\Lr{
    \begin{bmatrix}
    0 \\ 1
    \end{bmatrix}
    \zeta_{\pm 1}
    \ \middle |\ 
    \zeta_{\pm 1}\in\ker(T\mp 1)
    }
    ,\quad & \lambda = \pm 1.
    \end{cases}
\]
Here, $L\ker(\tilde T -\lambda) = d^*\ker(T-\lambda)+d^*(\lambda + S_o)\ker(T+\lambda)$ holds.
For any $\zeta_{\pm\lambda}\in\ker(T\mp \lambda)$,
\begin{align*}
    \inpr{d^*\zeta_{\lambda}}{(\lambda + S_o)d^*\zeta_{-\lambda}}
    =
    \inpr{\zeta_{\lambda}}{(\lambda + T)\zeta_{-\lambda}}=0.
\end{align*}
Thus, we have $L\ker(\tilde T -\lambda) = d^*\ker(T-\lambda)\bigoplus d^*(\lambda + S_o)\ker(T+\lambda)$, and the proof of the lemma is complete.
\end{proof}

\subsubsection{Birth eigenspaces}
We now consider $\Psi\in\ker(H-\lambda)\cap\mathcal{B}$ case.
\begin{lemma}
\label{lem:ev_birth}
The following holds.
\[
    \ker(H-\lambda)\cap\mathcal{B}
    =
    \begin{cases}
    \{0\},\quad & \lambda \neq \pm 1,
    \\
    \mathcal{B}_{\pm},\quad & \lambda = \pm 1.
    \end{cases}
\]
Note that $\mathcal{B}_{\pm} = \ker(C+1)\cap\ker(S_o\mp 1)$ was already defined.
\end{lemma}
\begin{proof}
Since $C=2d^*d-1$, we have $\ker d=\ker(C+1)$.
If $\Psi\in\ker(H-\lambda)\cap\mathcal{B}$, then
\[
    H\Psi = 
    \frac{1}{2}(CS_oC+S_o)\Psi = \frac{1}{2}(-C+1)S_o\Psi=(1+d^*d)S_o\Psi=S_o\Psi=\lambda\Psi.
\]
Here, $S_o^2=1$ means $\sigma(S_o) = \{-1,\, +1\}$, so the above equation shows $\Psi=0$ if $\lambda \neq \pm 1$,\ $\Psi\in\ker(S_o\mp 1)$ if $\lambda = \pm 1$.
Thus $\ker(H-\lambda)\cap\mathcal{B}\subset \ker(C+1)\cap\ker(S_o\mp 1)$ holds.
Inversely, if $\Psi \in \ker(C+1)\cap\ker(S_o\mp 1)$, then we can easily check $\Psi\in\ker(H\mp 1)$.
Therefore $\ker(H-\lambda)\cap\mathcal{B}\supset \ker(C+1)\cap\ker(S_o\mp 1)$ holds.
\end{proof}

\subsection{Continuous spectrum}
In this section we consider the continuous spectrum of $H$.
Since $H$ is a self-adjoint operotr, $\lambda\in\sigma(H)$ if and only if there exists a sequence $\{\Psi_n\}_{n\in\mathbb{N}}\subset \mathcal{A}$ satisfying $\|\Psi_n\|^2=1$ and $\|(H-\lambda)\Psi_n\|^2\to 0\ (n\to\infty)$.
\begin{lemma}
\label{lem:cont_lem1}
For $\lambda\in\sigma(H)$ and a sequence $\{\Psi_n\}_{n\in\mathbb{N}}\subset \mathcal{A}$ satisfying $(H-\lambda)\Psi_n=o(1)$ and $\|\Psi_n\|^2=1$ for any $n\in\mathbb{N}$, if $\Psi_n$ satisfies both conditions that $\lim_{n\to\infty}d\Psi_n=0$ and $\lim_{n\to\infty}dS_o\Psi_n=0$, then $\lambda\in\{-1, +1\}$ holds.
\end{lemma}
\begin{proof}
Since $C=2d^*d-1$ and $H=\frac{1}{2}(S_o+CS_oC)$, the assumption gives that
\begin{align}
    \nonumber
    \inpr{H\Psi_n}{\lambda\Psi_n} &=
    \inpr{\frac{\lambda}{2}(S_o + (2d^*d-1)S_o(2d^*d-1))\Psi_n}{\lambda\Psi_n}
    \\ \nonumber
    &=
    \inpr{\frac{\lambda}{2}(S_o -(2d^*d-1)S_o)\Psi_n}{\lambda\Psi_n} + o(1)
    \\
    &=
    \lambda^2 \inpr{S_o\Psi_n}{\Psi_n} + o(1).
    \label{eq:cont_lem1_1}
\end{align}
On the other hands, we have
\begin{align}
\label{eq:cont_lem1_2}
    \inpr{H\Psi_n}{\lambda\Psi_n}
    =
    \inpr{H\Psi_n}{H\Psi_n}
    +
    \inpr{H\Psi_n}{(\lambda - H )\Psi_n}=
    \inpr{S_o\Psi_n}{S_o\Psi_n} + o(1)=1+o(1).
\end{align}
From \eqref{eq:cont_lem1_1} and \eqref{eq:cont_lem1_2}, the following holds.
\[
    \lambda^2|\inpr{S_o\Psi_n}{\Psi_n}| = 1 + o(1).
\]
Here, the Cauchy–Schwarz inequality shows $|\inpr{S_o\Psi_n}{\Psi_n}|\leq 1$, so we have $|\lambda|\geq 1$.
However, we see $|\lambda|\leq 1$ from Corollary \ref{cor:radius_of_H}, thus $\lambda\in\{-1,+1\}$.
\end{proof}
\begin{lemma}
\label{lem:cont_lem2}
Following two conditions hold.
\begin{align*}
    (i)\quad &\sigma(H)\setminus\{-1,+1\}\subset \sigma(T)\cup\sigma(-T).
    \\
    (ii)\quad &(\sigma(T)\cup\sigma(-T))\setminus\{-1,+1\}\subset \sigma(H).
\end{align*}
In particular, combining (i) and (ii) shows $\sigma(H)\setminus\{-1,+1\} = (\sigma(T)\cup\sigma(-T))\setminus\{-1,+1\}$.
\end{lemma}
\begin{proof}
First of all, we will show (i).
Suppose $\lambda\in\sigma(H)\setminus\{-1,+1\}$ and there exists a sequence $\{\Psi_n\}_{n\in\mathbb{N}}$ satisfying $\|(H-\lambda)\Psi_n\| = o(1)$ and $\|\Psi_n\|^2=1$.
Then, we have.
\begin{align*}
    \|d(H-\lambda)\Psi_n\| = \|\tfrac{1}{2}d(S_o+CS_oC)\Psi_n-\lambda d\Psi_n\| = \|\tfrac{1}{2}dS_o(C+1)\Psi_n -\lambda d\Psi_n\|=o(1).
\end{align*}
Here, we consider $\lim_{n\to\infty}d\Psi_n \neq 0$ case.
Let $f_n =d\Psi_n$, and we can take a subsequence sequence $\{f_{n_k}\}$ such that $\inf_{k}\|f_{n_k}\|=:c>0$ holds.
From the above equation, we have
\begin{align*}
    \|(T-\lambda ) f_n\| = 
    \|(dS_od^* -\lambda)f_n\| = \| dS_od^*d\Psi_n - \lambda d\Psi_n\| 
    = \|\tfrac{1}{2}d S_o(C+1)\Psi_n -\lambda d\Psi_n\| = o(1).
\end{align*}
Let $\tilde f_{k} = f_{n_k}/\|f_{n_k}\|$, then we see $\|\tilde f_{n_k}\|^2=1$, and that
\[
    \|(T-\lambda)\tilde f_{k}\| \leq \frac{1}{c}\|(T-\lambda)\tilde f_{k}\| = o(1).
\]
Therefore, by taking a subset $X_1\subset\sigma(H)\setminus\{-1, +1\}$ whose any element $\lambda$ satisfying $\lim_{n\to\infty}d\Psi_n\neq 0$ with $\|(H-\lambda)\Psi_n\| = o(1)$, we obtain that $\sigma(H)\setminus\{-1, +1\}\cap X_1 \subset \sigma(T)$.
Next, we consider $\lim_{n\to\infty}d\Psi_n=0$ case.
In this case, we remark that $\lim_{n\to\infty}S_o d\Psi_n\neq 0$ because of Lemma \ref{lem:cont_lem1} and $\lambda\in\sigma(H)\setminus\{-1, +1\}$.
Let $g_n =dS_o\Psi_n$, then we have
\begin{align}
    \nonumber
    \|(T+\lambda)g_n\| = \|dS_od^* dS_o\Psi_n + \lambda dS_o\Psi_n\| &= \|\tfrac{1}{2}dS_o(C+1)S_o\Psi_n + \lambda dS_o\Psi_n\| 
    \\
    &= \|\tfrac{1}{2}dS_oCS_o\Psi_n + \tfrac{1}{2}d\Psi_n +\lambda dS_o\Psi_n\|.
    \label{eq:cont_lem2_1}
\end{align}
Here, we see that
\begin{align}
    \nonumber
    \|dS_oC(H-\lambda)\Psi_n\| &= \|\tfrac{1}{2}dS_oC(S_o + CS_o C)\Psi_n -\lambda dS_oC\Psi_n\|
    \\
    \nonumber
    &=\|\tfrac{1}{2}dS_oCS_o\Psi_n + \tfrac{1}{2}d\Psi_n -\lambda dS_o(2d^*d-1)\Psi_n\|
    \\
    \label{eq:cont_lem2_2}
    &=\|\tfrac{1}{2}dS_oCS_o\Psi_n + \tfrac{1}{2}d\Psi_n +\lambda dS_o\Psi_n\| + o(1).
\end{align}
The assumption says that $\|dS_oC(H-\lambda)\Psi_n\|=o(1)$ holds, so we obtain $\|(T+\lambda)g_n\|=o(1)$ by combining \eqref{eq:cont_lem2_1} and \eqref{eq:cont_lem2_2}.
We are now ready to take a subsequence $\{g_{n_k}\}$ such that $\inf_k\|g_{n_k}\|:=c>0$ holds. Let $\tilde g_{k} = g_{n_k}/\|g_{n_k}\|$, then we see $\|\tilde g_k\|^2=1$, and that
\[
    \|(T+\lambda)\tilde g_k\| \leq \frac{1}{c}\|(T+\lambda)\tilde g_{k}\| = o(1).
\]
Let a subset $X_2\subset \sigma(H)\setminus\{-1, +1\}$ whose any element $\lambda$ satisfying $\lim_{n\to\infty}d\Psi_n=0$ with $\|(H -\lambda)\Psi_n\| = o(1)$.
Above equation says that $\sigma(H)\setminus\{-1, +1\}\cap X_2 \subset \sigma(-T)$.
Since $X_1 \cup X_2 = \sigma(H)\setminus\{-1, +1\}$, we conclude $\sigma(H)\setminus\{-1, +1\}\subset \sigma(T)\cup\sigma(-T)$.

Secondly, we will show the statement (ii) of the lemma.
We suppose $\lambda\in\sigma(\pm T)\setminus\{-1,+1\}$ and there exists a sequence $\{h_n^\pm\}_{n\in\mathbb{N}}\subset \mathcal{V}$ such that $\|(T \mp \lambda) h_n^\pm\| = o(1)$ and $\|h_n^\pm\|=1$.

Since $\{h_n^+\}$ can always be given as a non-zero sequence and $\ker(d^*)=\{0\}$, we can take $\{h_n^+\}$ satisfying $\|d^*h_n^+\|\neq 0$.
Moreover, if $\|h_n^-\|\neq 0$, then $\lambda\not\in\{-1, +1\}$ says $\|(S_o+\lambda)d^*h_n^-\|\neq 0$, because
\[
    \|(S_o+\lambda)d^* h_n^-\|^2 = \inpr{d(S_o+\lambda)^2d^*h_n^-}{h_n^-} = \inpr{(1+\lambda^2 + 2T)h_n^-}{h_n^-} = (1-\lambda^2)\|h_n^-\|^2.
\]
From these arguments, we can define $\Psi^+_n = d^*h_n^+/\|d^*h_n^+\|$ and $\Psi^-_n =(S_o+\lambda )d^*h_n^-/\|(S_o + \lambda )d^*h_n^-\|$.
Here, we have
\begin{align*}
    \|(H-\lambda)\Psi_n^+\| &= \|\tfrac{1}{2}(S_o +CS_oC)d^*h_n^+ -\lambda d^*h_n^+ \|/\|d^*h_n^+\|
    \\
    &=\|\tfrac{1}{2}(C+1)S_od^*h_n^+ -\lambda d^*h_n^+ \|/\|d^*h_n^+\|
    \\
    &=\|dS_od^*h_n^+ -\lambda d^*h_n^+ \|/\|d^*h_n^+\|
    \\
    &=\|d^*(T-\lambda)h_n^+\|/\|d^*h_n^+\|
    \\
    &= o(1).
\end{align*}
Therefore $\sigma(T)\subset\sigma(H)$ holds.
Next, we have
\begin{align*}
    \|(H-\lambda)\Psi_n^-\| &=
    \|\tfrac{1}{2}(S_o + CS_oC)(S_o + \lambda)d^* h_n^- -\lambda(S_o+\lambda)d^*h_n^-\|/\|\lambda(S_o+\lambda)d^*h_n^-\|
    \\
    &=
    \|\tfrac{1}{2}(1+CS_oCS_o)d^*h_n^- + \tfrac{\lambda}{2}(C+1)S_od^*h_n^--\lambda(S_o+\lambda)d^*h_n^-
    \|/\|\lambda(S_o+\lambda)d^*h_n^-\|
    \\
    &=
    \|\tfrac{1}{2}\Lr{1+(2d^*d-1)S_o(2d^*d-1)S_o}d^*h_n^- + \lambda d^* Th_n^- -\lambda(S_o+\lambda)d^*h_n^-
    \|/\|\lambda(S_o+\lambda)d^*h_n^-\|
    \\
    &=\|(2d^*T^2-S_od^*T)h_n^- + \lambda d^*T h_n^- -\lambda(S_o+\lambda)d^*h_n^-
    \|/\|\lambda(S_o+\lambda)d^*h_n^-\|
    \\
    &=
    \|
    (-S_od^* + d^*(2T-\lambda))(T+\lambda)h_n^-
    \|/\|\lambda(S_o+\lambda)d^*h_n^-\|
    \\
    &= o(1).
\end{align*}
Thus, we cocnlude $\sigma(-T)\setminus\{-1, +1\}\subset\sigma(H)$.
Combining it with $\sigma(T)\subset\sigma(H)$ shows the statement of (ii).
\end{proof}
\begin{lemma}
\label{lem:cont_lem3}
For $K=\sigma_p(H)\cap\{-1, +1\}$, the following condition hold.
\begin{align*}
    \sigma(H)\setminus K = (\sigma(T)\cup\sigma(-T))\setminus K
\end{align*}
In particular, we have $\sigma(H) = (\sigma(T)\cup\sigma(-T))\cup\sigma_p(H)$ by considering the union set of the above and $\sigma_p(H)$.
\end{lemma}
\begin{proof}
If $\sigma_c(H)\cap\{-1, +1\} = \emptyset$, then $K=\sigma(H)\cap\{-1, +1\}$ and Lemma \ref{lem:cont_lem2} says that $\sigma(H)\setminus K = \sigma(H)\setminus\{-1, +1\} = (\sigma(T)\cup\sigma(-T))\setminus\{-1, +1\} \subset \sigma(T)\cup\sigma(-T)$.
Thus, we now consider only $\sigma_c(H)\cap\{-1, +1\} \neq \emptyset$ case.
For any $x\in\sigma_c(H)\cap\{-1, +1\}$, since $x\in\sigma_c(H)$, there exists a sequence $\{x_n\}_{n\in\mathbb{N}}\subset \sigma_c(H)\setminus\{-1, +1\}$ such that $\lim_{n\to\infty}x_n = x$.
Moreover, by Lemma \ref{lem:cont_lem2}, $x_n\in \sigma(T)\cup\sigma(-T)$ holds.
Here, $x\in\sigma(T)\cup \sigma(-T)$ also holds because $\sigma(T)\cup\sigma(-T)$ is a closed set.
Then, we have $\sigma_c(H)\cap\{-1, +1\}\subset \sigma(T)\cup\sigma(-T)$ and $\sigma(H)\setminus K \subset (\sigma(T)\cup\sigma(-T))\setminus K$ because of
\[
    \sigma(H) = (\sigma(H)\setminus\{-1, +1\}) \cup (\sigma_c(H)\cap\{-1, +1\}) \cup K
    \ \subset\  (\sigma(T)\cup\sigma(-T))\cup K.
\]
Moreover, by using a same method, we can easily show that $(\sigma_c(T)\cup\sigma_c(-T))\cap\{-1, +1\}\subset \sigma(H)$.
Since we have already proved that $\sigma_p(H) = \sigma_p(T)\cup\sigma_p(-T)\cup\{-1\}^{\dim \mathcal{B}_{-}}\cup\{+1\}^{\dim \mathcal{B}_{+}}$ in in Section \ref{subsec:point_sp},
so $\sigma_p(T)\cup\sigma_p(-T) \subset \sigma_p(H)$ holds and we have
\begin{align*}
    \sigma(T)\cup\sigma(-T) &\subset
    ((\sigma(T)\cup\sigma(-T))\setminus\{-1, +1\}) \cup 
    ((\sigma_c(T)\cup\sigma_c(-T))\cap\{-1, +1\}) \cup
    (\sigma_p(T)\cup\sigma_p(-T))
    \\
    &\subset \sigma(H).
\end{align*}
Thus $(\sigma(T)\cup\sigma(-T))\setminus K \subset \sigma(H)\setminus K$ holds, and the proof is completed.
\end{proof}



\printbibliography

\end{document}